\documentclass{article}

\usepackage{latexsym}
\usepackage{amssymb}
\usepackage{amsfonts}
\usepackage{amstext}
\usepackage{multicol}
\usepackage{multirow}
\usepackage{enumerate}
\usepackage{amsmath}
\usepackage{graphicx, psfrag}
\usepackage{array}
\usepackage{mathrsfs}
\usepackage{url}
\usepackage{amsthm}
\usepackage{esint}

\newcommand{\eps}{\varepsilon}
\newcommand{\R}{\mathbb{R}}
\newcommand{\C}{\mathbb{C}}
\newcommand{\N}{\mathbb{N}}

\newcommand{\D}{\partial}
\renewcommand{\Im}{\text{Im }}
\renewcommand{\Re}{\text{Re }}
\newtheorem{thm}{Theorem}[section]
\newtheorem{prop}{Proposition}[section]

\numberwithin{equation}{section}
\author{Matthew Masarik\\
University of Michigan\\
\texttt{masarikm@umich.edu}}
\title{The Wave Equation in a General Spherically Symmetric Particlelike Geometry}
\date{\today}

\begin{document}
\maketitle
\begin{abstract}We consider the Cauchy problem with smooth and compactly supported initial data for the wave equation in a general class of spherically symmetric geometries which are globally smooth and asymptotically flat. Under certain mild conditions on the far-field decay, we show that there is a unique globally smooth solution which is compactly supported for all times and \emph{decays in $L^{\infty}_{\text{loc}}$ as $t$ tends to infinity}. Because particlelike geometries are singularity free, they impose additional difficulties at the origin. Thus this study requires ideas and techniques not present in the study of wave equations in black hole geometries. We obtain as a corollary that solutions to the wave equation in the geometry of particle-like solutions of the SU(2) Einstein/Yang-Mills equations decay as $t\to \infty$.\end{abstract}
\section{Introduction}Recently there has been much interest in obtaining decay results for the wave equation in various \emph{black hole} geometries. In the case of the Schwarzschild metric, Kronthaler showed in \cite{thecauchyproblemforthewaveequationintheschwarzschildgeometry} that there exists a unique global solution to this problem and, moreover, the solution decays pointwise as $t \to \infty$. In \cite{onpointwisedecayoflinearwavesonaSchwarzschildblackholebackground} Donninger, et al. obtain the specific decay \emph{rate} $t^{-3}$ for solutions; and in \cite{decayratesforsphericalscalarwavesintheSchwarzschildgeometry} Kronthaler obtains the same rate under the assumption that the data is spherically symmetric (i.e. for the first angular mode of the full solution), along with the additional result that if the data is momentarily static (i.e., $\left.\D_t\phi\right|_{(x,0)} = 0$), then the decay rate can be improved to $t^{-4}$ (again for the first angular mode). For the Kerr metric (without a smallness restriction to the angular momentum), Finster et al. showed decay of solutions of the wave equation in \cite{decayofsolutionsofthewaveequationinthekerrgeometry}. If the case of sufficiently small angular momentum, Dafermos and Rodnianski were able to demonstrate the uniform boundedness of solutions to the wave equation in \cite{aproofoftheuniformboundednessofsolutionstothewaveequationsonslowlyrotatingkerrbackgrounds}, and Andersson and Blue obtained decay \emph{rates} in \cite{hiddensymmetriesanddecayforthewaveequationonthekerrspacetime}. However, in this paper we intend to study \emph{particle-like} geometries (i.e. non-singular and asymptotically flat). This is, therefore, an entirely novel problem.

We consider a 4-dimensional Riemannian manifold $\mathscr{M}$ with metric $g$: $(\mathscr{M},g)$ where the metric $g$ is given by
\begin{equation}ds^2 = g_{ij}dx^idx^j = -T^{-2}(r)dt^2 + K^2(r)dr^2 + r^2(d\theta^2 + \sin^2\theta d\phi^2),\label{metric}\end{equation}
and where $r\geq 0$, $0\leq \theta\leq \pi$, and $0\leq \phi \leq 2\pi$. We assume that the metric coefficients are globally smooth: $T,K \in C^{\infty}[0,\infty)$; we also assume that the metric is not degenerate: $T,K > 0$ (note that since we shall assume $T,K \to 1$ as $r \to \infty$, this implies $T$ and $K$ are bounded away from zero). We further assume
\begin{equation}K(0) = 1,\label{firstcondn}\end{equation}
\begin{equation}T'(0) = K'(0) = 0,\end{equation}
\begin{equation}T(r) \sim 1 + O\left(\frac{1}{r}\right) \text{ and } K(r) \sim 1 + O\left(\frac{1}{r}\right) \text{ as } r\to \infty,\label{TKto1}\end{equation}
and finally
\begin{equation}\frac{T'(r)}{T(r)} + \frac{K'(r)}{K(r)} \sim O\left(\frac{1}{r^2}\right) \text{ as } r \to \infty.\label{TKasymptotics}\end{equation}
In other words, we are assuming that near the origin, the $t =$ const. hyperplanes are similar to the Euclidean space $\R^3$ up to order $r^2$, and in the far-field limit $\mathscr{M}$ is the Minkowski space $\R^{1+3}$ up to order $r^{-1}$. These assumptions are essential in describing what we consider a particlelike geometry. The assumption $\eqref{TKasymptotics}$ is equivalent to assuming that $\frac{d}{dr}\log(TK) = O\left(\frac{1}{r^2}\right)$ for large $r$, so we assume control on the rate at which the log of $TK$ tends to 0. These assumptions are satisfied for the important examples of particle-like geometries (e.g., Minkowski, particle-like solutions of Einstein/Yang-Mills (EYM) with gauge group SU(2), c.f. \cite{existenceofinifinitelymanysmoothstaticglobalsolutionsoftheeinsteinyangmillsequations}).

We propose to study the Cauchy problem for the wave equation in this geometry. However, there is a boundary at $r=0$ and we must impose a boundary condition there. When considering black hole solutions, one merely requires that the data be compactly supported away from the horizon. Then, one can show that the solution never reaches the boundary, so that the natural boundary conditions are that the solution is zero at the horizon and at infinity. In the particle-like case, we must take a different approach, however, since there is no reason why the solution of the wave equation in a particle-like geometry should be always supported away from the origin. So we must determine the proper (i.e. physical) boundary condition at the origin. To do this, we recast this as a problem in Cartesian coordinates. Making this change of coordinates, the metric (in coordinates $(t,x,y,z)$) becomes 
\begin{equation}ds^2 = g_{ij}dx^idx^j,\end{equation}
where the nonzero metric coefficients are given by
\begin{align}&g_{11} = -T^{-2}(r)\notag\\
&g_{22} = \frac{x^2K^2(r) + y^2 + z^2}{r^2}\notag\\
&g_{24} = g_{42} = \frac{xz(K^2(r) - 1)}{r^2}\notag\\
&g_{23} = g_{32} = \frac{xy(K^2(r) - 1)}{r^2}\\
&g_{33} = \frac{x^2 + y^2K^2(r) + z^2}{r^2}\notag\\
&g_{34} = g_{43} = \frac{yz(K^2(r) - 1)}{r^2}\notag\\
&g_{44} = \frac{x^2 + y^2 + z^2K^2(r)}{r^2},\notag\end{align}
and $r = \sqrt{x^2 + y^2 + z^2}$. Note that these coefficients are globally smooth (the conditions $K(0) = 1, K'(0) = 0$ guaranteeing smoothness at the origin). This is obvious for each term except the diagonal $g_{ii}$ terms. Consider, for example, $g_{22}$. We can write
\[g_{22} = \frac{1}{r^2 }\left(\frac{x^2}{K^2} + y^2 + z^2\right) = 1 + \frac{x^2}{r^2K^2}\left(1 - K^2\right),\]
from which we see that $g_{22}$ is globally smooth. Similar arguments demonstrate the smoothness of the other diagonal terms. The wave equation in this geometry is given by 
\begin{equation}0= g^{ij}\nabla_i \nabla_j \zeta = \frac{1}{\sqrt{-g}}\frac{\D}{\D x^i}\left(\sqrt{-g}g^{ij}\frac{\D}{\D x^j}\right)\zeta =:\square \zeta.\label{waveequation}\end{equation}
Since this is a Lorentzian metric, the Laplacian will be a hyperbolic operator and we therefore expect finite speed of propagation. This coupled with compactly supported initial data suggests the asymptotic boundary condition $\zeta(t,x,y,z) \to 0$ as $r \to \infty$. We therefore study the Cauchy problem\footnote{We use the compact form $(\psi,i\psi_t)$ for the data in what follows, since this is most convenient when we reformulate this as a Hamiltonian problem later.}
\begin{equation}\begin{cases}\square \zeta(t,x,y,z) = 0, (x,y,z) \in \R^3, t >0 \\ (\zeta, i\zeta_t)(0,x,y,z) = Z_0(x,y,z) \in C_0^{\infty}(\R^3)^2.\end{cases}\label{cp1}\end{equation}
(We omit the asymptotic boundary condition at infinity since we will show that it is necessarily satisfied by the solution of $\eqref{cp1}$.) Next we write out explicitly $\square \zeta = 0$ in Cartesian coordinates:
\begin{equation} \zeta_{tt} = \sum_{i,j = 1}^3 a_{ik} \zeta_{x_ix_j} + \sum_{i=1}^3b_i\zeta_{x_i},\label{waveequation2}\end{equation}
where the coefficients are given by
\begin{align}&a_{11} = \frac{1}{r^2 T^2}\left(\frac{x^2}{K^2} + y^2 + z^2\right),\notag \\
&a_{22} =  \frac{1}{r^2 T^2}\left(x^2 + \frac{y^2}{K^2} + z^2\right),\notag \\
&a_{33} =  \frac{1}{r^2 T^2}\left(x^2 + y^2 + \frac{z^2}{K^2} \right),\notag \\
&a_{12} = a_{21} = \frac{(1-K^2)xy}{r^2T^2K^2}, \\
&a_{13} = a_{31} = \frac{(1-K^2)xz}{r^2T^2K^2},\notag \\
&a_{23} = a_{32} = \frac{(1-K^2)yz}{r^2T^2K^2}, \notag \\
&b_i = \left[ \frac{2}{r^2K^2}(1 - K^2) - \frac{1}{rK^2}\left(\frac{T'}{T} + \frac{K'}{K}\right)\right].\notag\end{align}
We will frequently suppress the arguments of functions to ease notation. We can show as before that these coefficients are globally smooth. If we now let $v = (\zeta_x,\zeta_y,\zeta_z,\zeta_t)^T$, then we can write equation $\eqref{waveequation2}$ as 
\begin{equation}A\D_t v - A_1\D_x v - A_2 \D_y v - A_3 \D_z v - Bv = 0,\label{hyperbolicsystem}\end{equation}
where 
\[A = \left(\begin{array}{cccc} a_{11} & a_{12} & a_{13} & 0\\ a_{21} & a_{22} & a_{23} & 0 \\ a_{31} & a_{32} & a_{33} & 0 \\  0 & 0 & 0 & 1\end{array}\right),\,\,
A_i = \left( \begin{array}{cccc} 0 & 0 & 0 & a_{i1} \\ 0 & 0 & 0 & a_{i2} \\ 0 & 0 & 0 & a_{i3} \\ a_{i1} & a_{i2} & a_{i3} & 0\end{array}\right),\,\,
\text{ and }
B  = \left( \begin{array}{cccc} 0 & 0 & 0 & 0 \\ 0 & 0 & 0 & 0 \\ 0 & 0 & 0 & 0 \\ b_1 & b_2& b_3& 0\end{array}\right).\]
Then, since the eigenvalues of $A$ are $1, T^{-2}, T^{-2},$ and $K^{-2}$ and these are all bounded away from zero, $A$ is uniformly positive definite and the system in $\eqref{hyperbolicsystem}$ is therefore a symmetric hyperbolic system (in the sense of section 5.3 in \cite{fritz}). Accordingly, there exists a unique, global, smooth solution that propagates with finite speed. Coupling this with the initial data yields a solution $\zeta$ of $\eqref{cp1}$ that is unique, smooth, globally defined, and compactly supported for each $t$.

We now wish to use this solution to understand the wave equation in the coordinates $(t,r,\theta,\phi)$. In particular, it is necessary to obtain a natural boundary condition to impose at $r = 0$. To this end, let us recall that $\frac{\D \zeta}{\D r} = \nabla \zeta \cdot \frac{(x,y,z)}{r}$, so that for any $\eps > 0$
\begin{align*}\int_{\D B(0,\eps)} \frac{\D \zeta}{\D r} dS &= \int_{B(0,\eps)}\Delta \zeta dV,\end{align*}
by the divergence theorem. Thus (using the notation that $\fint_{\Omega} f dx = \frac{1}{|\Omega|}\int_{\Omega} f dx$),
\begin{equation}\fint_{\D B(0,\eps)} \frac{\D \zeta}{\D r} dS = \frac{\eps}{3}\fint_{B(0,\eps)}\Delta \zeta dV.\label{bcint}\end{equation}
Since $\zeta$ is smooth, the integral on the right remains uniformly bounded as $\eps \searrow 0$, and thus $\eqref{bcint}$ yields 
\begin{equation*}0 = \lim_{\eps\searrow 0} \fint \frac{\D \zeta}{\D r} dS = \left.\frac{\D \zeta}{\D r}\right|_{r = 0}.\end{equation*}
Thus we obtain the boundary condition at $r = 0$:
\begin{equation}\left.\frac{\D \zeta}{\D r}\right|_{r = 0} = 0.\end{equation}
$\left.\frac{\D \zeta}{\D r}\right|_{r=0} = 0$. Now the wave equation in the coordinates $(t,r,\theta,\phi)$ reads (we now consider $\zeta = \zeta(t,r\,\theta,\phi)$)
\begin{equation}-T^2\zeta_{tt} + \frac{1}{K^2}\zeta_{rr} + \left(\frac{2}{K^2r} - \frac{1}{K^2}\left(\frac{T'}{T} + \frac{K'}{K}\right)\right)\zeta_r + \frac{\Delta_{S^2}}{r^2}\zeta = 0.\label{waveequation3}\end{equation}
We are therefore interested in solving the Cauchy problem
\begin{equation}\begin{cases} -T^2\zeta_{tt} + \frac{1}{K^2}\zeta_{rr} + \left(\frac{2}{K^2r} - \frac{1}{K^2}\left(\frac{T'}{T} + \frac{K'}{K}\right)\right)\zeta_r + \frac{\Delta_{S^2}}{r^2}\zeta = 0 \text{ on } \R \times (0,\infty)\times S^2 \\ \left.\frac{\D \zeta}{\D r}\right|_{r=0} = 0 \\ (\zeta,i\zeta_t)(0,r,\theta,\phi) = Z_0(r,\theta,\phi) \in\mathscr{A}^2,\end{cases}\label{cp2}\end{equation}
where
\begin{equation}\mathscr{A}:= \left\{ \psi \in C^{\infty}([0,\infty)\times S^2) : \left.\psi_r\right|_{r = 0} = 0 \text{ and there exists } R>0 \text{ so that } \psi(r,\theta,\phi) \equiv 0
\text{ for } r> R\right\}.\label{defofA}\end{equation}

\begin{thm}The Cauchy problem $\eqref{cp2}$ has a globally smooth, unique solution $\zeta$. Moreover,  $\zeta \in \mathscr{A}$ for each time $t$.\end{thm}
\begin{proof}This follows at once from the above result when we change to spherical coordinates and consider $\zeta = \zeta(t,r,\theta,\phi)$.\end{proof}

Let us now define the coordinate $u = u(r)$ by
\begin{equation}u(r) = -\int_{r}^{\infty} \frac{K(r')T(r')}{(r')^2}dr',\label{ucoordinate}\end{equation}
which maps the interval $(0,\infty)$ to $(-\infty,0)$. We record some asymptotics of $u$ which will be useful later. For large $r$, we have $u\nearrow 0$ according to
\begin{equation}u(r) = -\frac{1}{r} + O\left(\frac{1}{r^2}\right), \text{ or } \frac{1}{u} = -r + O(1).\label{smalluasymptotics}\end{equation}
For $r$ small we have $u \to -\infty$ according to 
\begin{equation}u(r) = -\frac{1}{r}+ O(1)\text{ or }\frac{1}{u} = -r + O(r^2).\label{smallrasymptotics}\end{equation}
If we let $\psi(t,u,\theta,\phi) = \zeta(t,r(u),\theta,\phi)$, then $\psi$ satisfies
\begin{equation}\left(-r^4\D_t^2 + \D_u^2 + \frac{r^2}{T^2}\Delta_{S^2}\right)\psi = 0\text{ on } \R \times (-\infty, 0) \times S^2.\label{waveequationinu}\end{equation}
Furthermore, $\psi$ is the unique, global, smooth solution of the Cauchy problem\footnote{The condition on $\psi$ as $u \to -\infty$ is precisely the condition $\left.\zeta_r\right|_{r = 0} = 0$ after changing variables and employing the asymptotics $\eqref{smalluasymptotics}$ and $\eqref{smallrasymptotics}$.} 
\begin{equation}\begin{cases}\left(-r^4\D_t^2 + \D_u^2 + \frac{r^2}{T^2}\Delta_{S^2}\right)\psi = 0 \text{ on } \R \times (-\infty,0) \times S^2\\ \psi_u = O\left(\frac{1}{u^3}\right) \text{ as } u\to -\infty \\(\psi,i\psi_t)(0,r,\theta,\phi) = \Psi_0(r,\theta,\phi) \in \mathscr{B}^2,\end{cases}\label{cp3}\end{equation}
where $\psi \in \mathscr{B}$ if and only if $\psi \in C^{\infty}((-\infty,0)\times S^2)$ and 
\begin{enumerate}[(i)]
\item there exists $u_0 < 0$ so that $\psi(u,\theta,\phi) \equiv 0$ for all $u>u_0$;
\item $\psi_u = O\left(\frac{1}{u^3}\right) \text{ as } u \to -\infty$; and 
\item $\psi$ and all its derivatives have finite limits as $u \to -\infty$.
\end{enumerate}

Observe that $\mathscr{B}$ and $\mathscr{A}$ are related to each other in the sense that, given $\zeta(r,\theta,\phi) \in \mathscr{A}$, $\psi(u,\theta,\phi) := \zeta(r(u),\theta,\phi) \in \mathscr{B}$. To see this, note that $\zeta_r = O(r)$ for small $r$, and thus we have
\begin{equation*}\psi_u = \zeta_r \frac{dr}{du} = \zeta_r \frac{r^2}{KT} = O(r^3)\end{equation*}
for small $r$. Owing to the asymptotics in $\eqref{smallrasymptotics}$ for small $r$, it follows that $\psi_u = O\left(\frac{1}{u^3}\right)$ as $u \to -\infty$. Recalling also that $\zeta$ is smooth up to the origin, it follows that $\psi$ and all the derivatives of $\psi$ have finite limits as $u \to -\infty$.

We also note that $\psi \in \mathscr{B}$ for all times $t$. This follows from the above observations and the fact that $\zeta \in \mathscr{A}$. Thus the energy 
\begin{equation}E(\psi) := \int_0^{2\pi}\int_{-1}^1\int_{-\infty}^0 r^4 (\psi_t)^2 + (\psi_u)^2 + \frac{r^2}{T^2}\left(\frac{1}{\sin^2\theta}(\D_{\phi}\psi)^2 + \sin^2\theta (\D_{\cos\theta}\psi)^2\right)dud(\cos\theta)d\phi\label{energy}\end{equation}
is well-defined. Moreover, the summability guarantees that we may compute $\frac{d}{dt}E(\psi)$ by differentiating under the integral.\footnote{We can do this since according to the asymptotics $\eqref{smalluasymptotics}$ and $\eqref{smallrasymptotics}$ the coefficients on the first and third terms in the integrand decay at least as fast as $\frac{1}{u^2}$ as $u\to -\infty$, and the $\psi_u$ term decays as $\frac{1}{u^3}$. Then, using that all the derivatives of $\psi$ have finite limits as $u\to-\infty$, we can apply Lebesgue's dominated convergence theorem to differentiate under the integral.} Integrating by parts and using the asymptotics $\eqref{smallrasymptotics}$ to account for the boundary terms yields that $\frac{d}{dt}E(\psi) = 0$; i.e. the energy is conserved.

We next let $\Psi = (\psi,i\psi_t)^T$ and recast $\eqref{cp3}$ as a Hamiltonian system; i.e. $\Psi$ is the unique global solution in $\mathscr{B}^2$ for all times $t$ of the Cauchy problem
\begin{equation}\begin{cases}i\D_t \Psi = H\Psi \text{ on } \R \times (-\infty,0) \times S^2 \\\Psi(0,u,\theta,\phi) = \Psi_0(u,\theta,\phi) \in \mathscr{B}^2,\label{cp4}\end{cases}\end{equation}
where the Hamiltonian $H$ is given by 
\begin{equation}H = \left(\begin{array}{cc} 0 & 1 \\ A & 0 \end{array}\right) \text{ and } A = -\frac{1}{r^4}\D_u^2 - \frac{\Delta_{S^2}}{r^2T^2}.\label{hamiltonian}\end{equation}

We can also see that the energy functional induces an inner product on $\mathscr{B}^2$. For $\Psi, \Gamma \in \mathscr{B}^2$, the inner product $\langle \Psi,\Gamma\rangle$ is given by
\begin{equation}\int_0^{2\pi}\int_{-1}^1 \int_{-\infty}^0 r^4 \psi_2\overline{\gamma_2} + (\D_u\psi_1)\overline{(\D_u\gamma_1)} + \frac{r^2}{T^2}\left(\frac{1}{\sin^2\theta}(\D_{\phi}\psi_1)\overline{(\D_{\phi}\gamma_1)} + \sin^2\theta (\D_{\cos\theta}\psi_1)\overline{(\D_{\cos\theta}\gamma_1)}dud(\cos\theta)\right)d\phi.\label{innerproduct}\end{equation}

\begin{prop}$H$ is symmetric with respect to $\langle \cdot,\cdot \rangle$ on $\mathscr{B}^2$.\end{prop}
\begin{proof}Consider arbitrary $\Psi_0 \in \mathscr{B}^2$. Corresponding to $\Psi_0$ is a solution $\Psi$ of $\eqref{cp4}$, and $E(\Psi) = \langle \Psi,\Psi \rangle$ is conserved. Thus we have
\begin{align*}0 &= \frac{d}{dt}\langle \Psi, \Psi \rangle \\&= \langle \D_t\Psi, \Psi \rangle + \langle\Psi, \D_t\Psi\rangle \\ &= -i\langle H\Psi, \Psi\rangle + i \langle \Psi, H\Psi\rangle.\end{align*}
This shows that $\langle H\Psi, \Psi\rangle = \langle \Psi, H\Psi\rangle$. Now this expression holds independent of $t$, and in particular it holds at $t = 0$. Thus $\langle H\Psi_0, \Psi_0\rangle = \langle \Psi_0, H\Psi_0\rangle$. A simple polarization argument then shows that since $\langle H\Psi_0,\Psi_0\rangle \in \R$ for each $\Psi_0 \in\mathscr{B}^2$, $H$ is indeed symmetric on $\mathscr{B}^2$.\end{proof}

We can reduce this from a three-dimensional problem to a one-dimensional problem by projecting our solution onto the spherical harmonics: 
\begin{equation}\Psi(t,u,\theta,\phi) = \sum_{l=0}^{\infty}\sum_{|m|\leq l}\Psi^{lm}(t,u)Y_{lm}(\theta,\phi)\label{sphericalharmonicdecomp},\end{equation}
where the $Y_{lm}$ are the spherical harmonics (i.e. $\Delta_{S^2}Y_{lm} = -l(l+1)Y_{lm}$) and this series converges uniformly and absolutely for fixed $(t,u) \in \R \times (-\infty,0)$ (c.f. \cite{courant1}). The inner product $\langle \cdot,\cdot\rangle$ decomposes as
\begin{align}\langle \Psi,\Gamma \rangle &= \sum_{l=0}^{\infty}\sum_{|m|\leq l} \langle \Psi^{lm},\Gamma^{lm}\rangle_l \notag\\
&=\sum_{l=0}^{\infty}\sum_{|m|\leq l} \int_{-\infty}^0 r^4 \psi^{lm}_2\overline{\gamma_2^{lm}} + (\D_u\psi^{lm}_1)\overline{(\D_u\gamma^{lm}_1)} + \frac{r^2}{T^2}l(l+1)\psi^{lm}_1\overline{\gamma^{lm}_1}du,\label{innerproduct2}\end{align}
and the action of the Hamiltonian decomposes as
\begin{equation}H\Psi = \sum_{l=0}^{\infty}\sum_{|m|\leq l}H_l \Psi^{lm}Y_{lm},\end{equation}
where 
\begin{equation}H_l = \left(\begin{array}{cc}0 & 1\\ A_l & 0\end{array}\right) \text{ and } A_l = -\frac{1}{r^4}\D_u^2 + \frac{l(l+1)}{r^2T^2}.\label{defofHl}\end{equation}
We note also that\footnote{This is different from the black hole case, since we must now account for the fact that our solution of the original problem $\eqref{cp2}$ need not be supported away from the origin} $H_l$ is symmetric on $\mathscr{C}_l^2$, where $\psi \in \mathscr{C}_l$ if and only if $\psi \in C^{\infty}(-\infty,0)$ and 
\begin{enumerate}[(i)]
\item there exists $u_0 < 0$ so that $\psi(u,\theta,\phi) \equiv 0$ for all $u>u_0$;
\item $\psi_u = O\left(\frac{1}{u^3}\right) \text{ as } u \to -\infty$;
\item $\psi$ and all its derivatives have finite limits as $u \to -\infty$;
\item if $l\neq 0$, $\psi = O\left(\frac{1}{u^2}\right)$ as $u\to -\infty$.
\end{enumerate}
The symmetry statement follows, since for $\Psi^{lm}, \Gamma^{lm} \in \mathscr{C}_l^2$ we have
\begin{align*}\langle H_l \Psi^{lm}, \Gamma^{lm}\rangle_l &= \langle H(\Psi^{lm}Y_{lm}), \Gamma^{lm}Y_{lm}\rangle \\ &= \langle \Psi^{lm}Y_{lm}, H(\Gamma^{lm}Y_{lm})\rangle \\ &= \langle \Psi^{lm}, H_l \Gamma^{lm}\rangle_l.\end{align*}
The component functions $\Psi^{lm}$ are global, smooth solutions in $\mathscr{C}_l^2$ for each time $t$ of the Cauchy problem
\begin{equation}\begin{cases}i\D_t \Psi^{lm} = H_l \Psi^{lm} \text{ on } \R \times (-\infty,0) \\ \Psi^{lm}(0,u) = \Psi_0^{lm}(u) \in \mathscr{C}^2.\label{cp5}\end{cases}\end{equation}
That the $\Psi^{lm}$ satisfy conditions (i) - (iii) has been demonstrated; we must still verify condition (iv). This, however, follows from the fact that if $\psi(r,\theta,\phi) = \sum_{l=0}^{\infty}\sum_{|m|\leq l}\psi^{lm}(r)Y_{lm}(\theta,\phi)$ and $\psi$ is well-defined at the origin, then $\psi^{lm}(0) = 0$ for $l\neq 0$. Indeed then, since our first solution $\zeta$ was well-defined at the origin and smooth up to the origin with $\D_r\zeta(t,0,\theta,\phi) = 0$, it follows that $\zeta^{lm} = O(r^2)$ near the origin for $l\neq 0$ ($\zeta^{lm}$ being the component functions in the spherical harmonic expansion of $\zeta$). Translating this in terms of the $u$ variable implies that, indeed, for $l\neq 0$, $\Psi^{lm} = O\left(\frac{1}{u^2}\right)$ as $u\to -\infty$. The symmetry of the Hamiltonian implies that the energy $E_l(\Psi^{lm}) := \langle \Psi^{lm},\Psi^{lm}\rangle_l$ is conserved for solutions of $\eqref{cp5}$, and energy conservation implies that $\Psi^{lm}$ are the unique solutions of $\eqref{cp5}$ in $\mathscr{C}_l^2$.

\section{Spectral Analysis \& The Hamiltonian}
We wish to derive a representiation formula for $\Psi^{lm}$. To that end, we wish to apply Stone's formula to $H_l$, which expresses the spectral projections of $H_l$ in terms of the resolvent. However, Stone's formula applies to self-adjoint operators, so we must find a self-adjoint extension of $H_l$. To that end, we must find a Hilbert space on which $H_l$ is densely defined. Let us first note that we can write 
\begin{equation*}\langle \Psi,\Gamma \rangle_l = \langle \psi_1, \gamma_1\rangle_{l_1} + \langle \psi_2, \gamma_2\rangle_{l_2},\end{equation*}
where, of course, $\langle \cdot,\cdot \rangle_{l_1}, \langle \cdot,\cdot \rangle_{l_2}$ correspond to the terms in the integral in $\eqref{innerproduct2}$ acting on the first and second components of the input functions, respectively.
Then we let 
\[\mathscr{H}_{r^2} := \left(\left\{\psi : r^2\psi \in L^2(-\infty,0)\right\},\langle\cdot,\cdot\rangle_{l_1}\right)\]
and 
\[\mathscr{H}^1_{V_l} := \left(\left\{\psi : \psi_u \in L^2(-\infty,0)\text{ and } r^2v_l^{\frac{1}{2}}\psi \in L^2(-\infty,0)\right\},\langle\cdot,\cdot\rangle_{l_2}\right),\]
where $V_l = \frac{l(l+1)}{r^2T^2}$. Then we take $\mathscr{H}_{r^2,0}, \mathscr{H}^1_{V_l,0}$ to be the completion of $\mathscr{C}_l$ within $\mathscr{H}_{r^2},\mathscr{H}^1_{V_l}$, respectively. Finally, we take $\mathscr{H} = \mathscr{H}_{r^2,0} \oplus \mathscr{H}^1_{V_l,0}$.

\begin{prop}The operator $H_l$ with domain $\mathscr{D}(H_l) = \mathscr{C}_l^2$ is essentially self-adjoint in the Hilbert space $\mathscr{H}$.\label{essentiallyselfadjoint}\end{prop}
\begin{proof}To prove this, we will use the following version of Stone's theorem (c.f. \cite{reedsimon1}, Sec. VIII.4):
\begin{thm}[{\bf Stone's Theorem}] Let $U(t)$ be a strongly continuous one-parameter unitary group on a Hilbert space $\mathscr{H}$. Then there is a self-adjoint operator $A$ on $\mathscr{H}$ so that $U(t) = e^{itA}$. Furthermore, let $\mathscr{D}$ be a dense domain which is invariant under $U(t)$ and on which $U(t)$ is strongly differentiable. Then $i^{-1}$ times the strong derivative of $U(t)$ is essentially self-adjoint on $\mathscr{D}$ and its closure is $A$.\label{stonestheorem}\end{thm}
Consider the Cauchy problem $\eqref{cp5}$. We have already demonstrated that there is a unique solution $\Psi^{lm}$ to this problem in $\mathscr{C}_l^2$ for each time $t$. Therefore we may define the operators
\begin{equation*}U(t): \mathscr{C}_l^2 \mapsto \mathscr{C}_l^2\text{ by } \end{equation*}
\begin{equation*} U(t)\Psi_0^{lm} = \Psi^{lm}(t).\end{equation*}
The energy conservation guarantees that the $U(t)$ are unitary on $\mathscr{C}_l^2$ with respect to $\langle \cdot, \cdot \rangle_l$ and they therefore extend to unitary operators on $\mathscr{H}$. The uniqueness of the solution guarantees that $U(0) = Id.$ and $U(t)U(s) = U(t+s)$ for all $t,s \in \R$, and thus the $U(t)$ form a one-parameter unitary group on $\mathscr{H}$.

We next wish to show that the $U(t)$ are strongly continuous on $\mathscr{H}$. Thus let $\Psi \in \mathscr{H}$. Then there exists $(\Psi_n) \subset \mathscr{C}_l^2$ such that $\Psi_n \to \Psi$ and we have 
\begin{equation*} \|U(t)\Psi - \Psi \| \leq \| U(t) \Psi - U(t) \Psi_n \| + \| U(t) \Psi_n - \Psi_n \| + \|\Psi_n - \Psi\|.\end{equation*}
Thus since the $U(t)$ are unitary and since $U(t)$ is obviously strongly continuous on $\mathscr{C}_l^2$, it follows that the $U(t)$ are strongly continuous on $\mathscr{H}$. Moreover, the smoothness of the solution guarantees that the $U(t)$ are strongly differentiable on $\mathscr{C}_l^2$. A simple calculation shows that for $(\psi_1,\psi_2)^T \in \mathscr{C}_l^2$,
\begin{equation*}\lim_{h\searrow 0}\frac{1}{h}\left(U(t)(\psi_1,\psi_2)^T - (\psi_1,\psi_2)^t\right) = (-i\psi_2, -iA_l\psi_1)^T = -iH_l(\psi_1,\psi_2)^T,\end{equation*}
and thus that $i^{-1}$ times the strong derivative of $U(t)$ is $-H_l$. 

Therefore, since $\mathscr{C}_l^2$ is invariant under $U(t)$, $H_l$ is essentially self-adjoint on $\mathscr{C}_l^2$. \end{proof}
Thus, $H_l$ has a unique self-adjoint extension $\bar{H}_l$ defined on a dense domain in $\mathscr{H}$ containing $\mathscr{C}_l^2$. The specifics of the domain, however, are irrelevant to our study, so we ignore these details.

\section{Stone's Formula \& The Resolvent}We now recall Stone's formula (c.f. \cite{reedsimon1}), since this is the tool by which we will derive a representation formula for $\Psi^{lm}$:
\begin{thm}[{\bf Stone's Formula}]For a self-adjoint operator $A$, the spectral projections are given by
\begin{equation}\frac{1}{2}\left(P_{[a,b]} + P_{(a,b)}\right) = \lim_{\eps\searrow 0} \frac{1}{2\pi i}\int_a^b \left[ (A - \omega - i\eps)^{-1} - (A - \omega + i\eps)^{-1}\right]d\omega,\end{equation}
where the limit is taken in the strong operator topology.\label{stoneformula}\end{thm}
So we see then that in order to utilize Stone's formula, we must study the resolvent of $\bar{H}_l$. To that end, we consider the eigenvalue equation
\begin{equation}\bar{H}_l \Gamma = \omega \Gamma.\label{eigenvalueequation}\end{equation}
Since $\bar{H}_l$ is self-adjoint on a domain in $\mathscr{H}$, it follows that the spectrum $\sigma(\bar{H}_l) \subset \R$ and that the resolvent $(\bar{H}_l - \omega)^{-1}: \mathscr{H} \mapsto \mathscr{H}$ exists for all $\omega \in \C\setminus\R$. Thus, the eigenvalue equation $\eqref{eigenvalueequation}$ has no solutions in $\mathscr{H}$ for $\Im \omega \neq 0$. However, $\eqref{eigenvalueequation}$ is equivalent to the ODE 
\begin{equation}-\gamma''(u) - \omega^2 r^4\gamma + \frac{r^2}{T^2}l(l+1)\gamma = 0 \text{ on } (-\infty,0),\label{gammaequation}\end{equation}
where the arguments are $r = r(u)$ and $T = T(r(u))$. We will construct the resolvent out of solutions to this ODE.

To solve this ODE, let us first note that if we consider the coordinate $s(u)$ given by
\begin{equation}s(u) = \int_{-\infty}^u r^2(u')du'\label{defofs}\end{equation}
and let 
\begin{equation}\eta(s) = r(u(s))\gamma(u(s)),\label{defofeta}\end{equation}
then $\eta$ solves the ODE
\begin{equation}-\eta''(s) - \omega^2\eta(s) + \left(\frac{l(l+1)}{r^2T^2} - \frac{1}{rT^2K^2}\left(\frac{T'}{T} + \frac{K'}{K}\right)\right)\eta = 0 \text{ on } (0,\infty).\label{etaode}\end{equation}
Let us note that we may regard $s$ as a function of $r$ by considering $s(u(r))$, which yields 
\begin{equation}s(u(r)) = \int_0^rK(r')T(r')dr'\label{defofs2}.\end{equation}

We now look to construct two linearly independent solutions of the ODE $\eqref{etaode}$, one satisfying boundary conditions at $s = 0$ and the other satisfying asymptotic boundary conditions at $s = \infty$. In what follows, we will let $\lambda = l + \frac{1}{2}$, so that $l(l+1) = \lambda^2 - \frac{1}{4}$. An excellent reference for what follows is \cite{potentialscattering} and we give them credit for the basic idea in finding these solutions.
\subsection{The Solution with Boundary Conditions at $s =0$}
We first consider the solution satisfying boundary conditions at $s = 0$: call this solution $\eta^1(\lambda, \omega,s)$. We shall require 
\begin{equation}\lim_{s\searrow 0}\eta^1(\lambda,\omega,s)s^{-\lambda - \frac{1}{2}} = 1.\label{eta1bc}\end{equation}
We will construct $\eta^1(\lambda,\omega,s)$ via a perturbation series, so let us define
\begin{equation}\eta^1_0(\lambda,\omega,s) = \left(\frac{2}{\omega}\right)^{\lambda}\Gamma(\lambda + 1)J_{\lambda}(\omega s) \text{ for } \omega \neq 0,\label{defofeta10}\end{equation}
where $\Gamma$ is the gamma function and $J_{\lambda}$ is the Bessel function of the first kind (c.f. \cite{watsonbesselfunctions} on Bessel functions and \cite{specialfunctions} on Bessel and Hankel functions). Then, we rewrite the ODE $\eqref{etaode}$ as 
\begin{equation}\eta''(s) + \left(\omega^2 - \frac{\lambda^2 - \frac{1}{4}}{s^2}\right)\eta(s) = \left(\left(\lambda^2 - \frac{1}{4}\right)\left[\frac{1}{r^2T^2} - \frac{1}{s^2} \right] - \frac{1}{rT^2K^2}\left(\frac{T'}{T} + \frac{K'}{K}\right)\right)\eta(s).\label{etaoderewrite}\end{equation}

The Green's function for the operator on the left-hand side of the above equation (satisfying zero boundary conditions at $s = 0$) is 
\begin{equation}G(\lambda,\omega,s,y) = H(s - y)\frac{1}{2\lambda}\left(\eta_0^1(\lambda,\omega,s)\eta^1_0(-\lambda,\omega,y) - \eta^1_0(-\lambda,\omega,s)\eta_0^1(\lambda,\omega,y)\right),\label{defofG}\end{equation}
where $H$ is the usual Heaviside function.

One then obtains the integral equation 
\begin{equation}\eta^1(\lambda,\omega,s) = \eta_0^1(\lambda,\omega,s) + \int_0^s G(\lambda,\omega, s,y)W(y)\eta^1(\lambda,\omega,y)dy,\label{eta1intequation}\end{equation}
where 
\begin{equation}W(y) = \left(\lambda^2 - \frac{1}{4}\right)\left(\frac{1}{r^2T^2} - \frac{1}{y^2}\right) + V(y), V(y) = -\frac{1}{rT^2K^2}\left(\frac{K'}{K} + \frac{T'}{T}\right).\label{defofW}\end{equation}
Note that since $W$ is integrable, a smooth solution of this integral equation will indeed be a solution of the ODE $\eqref{etaode}$ with boundary conditions $\eqref{eta1bc}$.

We next show that $W(y)$ is integrable. To that end, note that $V(s) = O(1)$ as $s \to 0$ (due to $K'(0) = 0 =T'(0)$) and that $V(s) = O\left(\frac{1}{s^3}\right)$ as $s\to \infty$, since $s \sim r$ for large $s$ and the condition $\eqref{TKasymptotics}$. To study the other term in $W$, first note that $s = T(0)r + O(r^3)$ for small $s$, and this implies that 
\begin{equation*}\left|\frac{1}{r^2T^2} - \frac{1}{s^2}\right| = O(1)\end{equation*}
for small $s$. Finally, for large $s$, we use the conditions $\eqref{TKto1}$ which imply that 
\begin{equation*}\left|\frac{1}{r^2T^2} - \frac{1}{s^2}\right| = O\left(\frac{\log s}{s^3}\right).\end{equation*}
Thus, altogether this yields that $W$ is $O(1)$ near the origin and decays like $\frac{\log s}{s^3}$ as $s\to\infty$, and so $\|W\|_{L^1(0,\infty)}< \infty$.

Now, in Appendix A of \cite{potentialscattering}, it is shown that for $0<y<s$ we have
\begin{equation}|G(\lambda,\omega,s,y)| \leq Ce^{|\Im \omega | (s-y)}\left(\frac{s}{1 + |\omega| s}\right)^{\lambda + \frac{1}{2}}\left(\frac{y}{1 + |\omega|y}\right)^{-\lambda + \frac{1}{2}},\label{Gbound}\end{equation}
for some $C > 0$ depending on $\lambda$. Then we write
\begin{equation}\eta^1(\lambda,\omega,s) = \sum_{n=0}^{\infty}\eta^1_n(\lambda,\omega,s)\label{eta1series}\end{equation}
for 
\begin{equation}\eta^1_{n}(\lambda,\omega,s) = \int_0^sG(\lambda,\omega,s,y)W(y)\eta^1_{n-1}(\lambda,\omega,y)dy.\label{eta1recurrence}\end{equation}
In the same appendix, it is shown that 
\begin{equation}|\eta^1_0(\lambda,\omega,s)|\leq Ce^{|\Im \omega|s}\left(\frac{s}{1 + |\omega|s}\right)^{\lambda + \frac{1}{2}}\label{eta10bound}\end{equation}
and it's then easy to show by induction and using $\eqref{Gbound}$, that 
\begin{equation}|\eta^1_n(\lambda,\omega,s)|\leq C e^{|\Im \omega|s} \frac{(C\cdot P(s))^n}{n!}\left(\frac{s}{1 + |\omega|s}\right)^{\lambda + \frac{1}{2}},\label{eta1nbound}\end{equation}
where
\begin{equation}P(s) = \int_0^s \frac{yW(y)}{1 + |\omega|y}dy\label{defofP}.\end{equation}

Thus the series $\eqref{eta1series}$ is bounded term-by-term by an exponential series and the following bounds are immediate:
\begin{equation}|\eta^1(\lambda,\omega,s)| \leq C e^{|\Im \omega| s}e^{CP(s)}\left(\frac{s}{1 + |\omega|s}\right)^{\lambda + \frac{1}{2}},\label{eta1bound}\end{equation}
as well as
\begin{equation}|\eta^1(\lambda,\omega,s) - \eta^1_0(\lambda,\omega,s)| \leq C e^{|\Im \omega| s}\left(e^{CP(s)} - 1\right)\left(\frac{s}{1 + |\omega|s}\right)^{\lambda + \frac{1}{2}}.\label{eta1minuseta10bound}\end{equation}
Also, since the series $\eqref{eta1series}$ is bounded term-by-term by an exponential series, it follows that it converges uniformly on compact sets in $s$.

We now check the smoothness of $\eta^1_0$. First let us control $\frac{\D G}{\D s}$. To that end, we observe that a short calculation gives
\begin{equation}\frac{d}{ds}\eta^1_0(\lambda,\omega,s) = \frac{1}{2s}\eta^1_0(\lambda,\omega,s) + \lambda \eta^1_0(\lambda - 1,\omega,s) - \frac{1}{\lambda + 1}\left(\frac{\omega}{2}\right)^2\eta^1_0(\lambda,\omega,s).\label{deta10}\end{equation}
Using the bound $\eqref{eta10bound}$ we can bound $\frac{d}{ds}\eta^1_0(\lambda,\omega,s)$:
\begin{equation}\left|\frac{d}{ds}\eta^1_0(\lambda,\omega,s)\right| \leq C e^{|\Im \omega|s}\left(\frac{s}{1 + |\omega|s}\right)^{\lambda - \frac{1}{2}},\label{deta1bound}\end{equation}
and $C$ is some constant depending on $\lambda$. We can then use this to bound $\frac{\D G}{\D s}$:
\begin{equation}\left| \frac{\D G}{\D s}(\lambda,\omega,s,y)\right| \leq C e^{|\Im \omega|(s+ y)}\left(\frac{s}{1 + |\omega|s}\right)^{\lambda + \frac{1}{2}}\left(\frac{y}{1 + |\omega|y}\right)^{-\lambda - \frac{1}{2}}.\label{dGbound}\end{equation}
Thus $\eqref{dGbound}$ and $\eqref{eta1bound}$ enable us to compute $\frac{d}{ds}\eta^1(\lambda,\omega,s)$ as 
\begin{equation}\frac{d}{ds} \eta^1(\lambda,\omega,s) = \frac{d}{ds}\eta^1_0(\lambda,\omega,s) + \int_0^s \frac{\D G}{\D s}(\lambda,\omega,s,y)W(y)\eta^1(\lambda,\omega,y)dy.\label{deta}\end{equation}
This yields 
\begin{equation}\left|\frac{d}{ds} \eta^1(\lambda,\omega,s) - \frac{d}{ds}\eta^1_0(\lambda,\omega,s)\right| \leq C e^{3|\Im \omega|s}\left(\frac{s}{1 + |\omega|s}\right)^{\lambda + \frac{1}{2}},\label{deta1minusdeta10bound}\end{equation}
since $P$ is bounded as $s \to \infty$ and $W$ is integrable. We can carry out a similar procedure to bound $\frac{\D^2G}{\D s^2}$ and conclude that $\eta^1 \in C^2(0,\infty)$, and thus, $\eta^1$ solves the ODE $\eqref{etaode}$ along with the boundary conditions $\eqref{eta1bc}$.

We also claim that $\eta^1$ is analytic in $\omega$ in the region $\omega \neq 0$; we will show this using Morera's theorem. So first note that $\eta^1_0(\lambda,\omega,s)$, for fixed $s\in (0,\infty)$, is analytic for $\omega \neq 0$. Assume that the same holds for $\eta^1_n(\lambda,\omega,s)$. Recall the definition $\eqref{eta1recurrence}$. It's easy to show the continuity of $\eta^1_{n+1}$ in $\omega$ using the dominated convergence theorem, the analyticity of $G$ in $\omega$, and the induction hypothesis. Then let $C$ be a closed contour in $\C\setminus\{0\}$ and consider
\begin{equation*}\int_C\eta^1_{n+1}(\lambda,\omega,s)d\omega = \int_C\int_0^sG(\lambda,\omega,s,y)W(y)\eta^1_n(\lambda,\omega,y)dyd\omega.\end{equation*}
Using the integrability of $W$ and the bounds $\eqref{Gbound}, \eqref{eta1nbound}$, we may interchange the order of integration, and the analyticity of $G$ and $\eta^1_n$ then yields that the integral is zero. Morera's theorem then guarantees that $\eta^1_{n+1}$ is analytic in $\omega$, and by induction, this holds for each $n\in \N$. Furthermore, the uniform convergence of the series $\eqref{eta1series}$ then yields that $\eta^1(\lambda,\omega,s)$ is analytic in $\omega$ in the region $\C\setminus\{0\}$ for fixed $s\in (0,\infty)$.

We note also that the only restriction on $\omega$ is that $\omega \neq 0$, but we claim that, in fact, $\eta^1(\lambda,\omega,s)$ can be extended continuously to $\omega = 0$. To this end, let us first demonstrate that the integral equation $\eqref{eta1intequation}$ has a unique solution. Indeed, suppose that $\tilde{\eta}$ is another solution of $\eqref{eta1intequation}$ and fix $s>0$. Since it must be that $\lim_{s\searrow 0}\tilde{\eta}(\lambda,\omega,s)s^{-\lambda - \frac{1}{2}} = 1$, there is a constant $C>0$ so that 
\begin{equation*}\tilde{\eta}(\lambda,\omega,y) \leq C e^{|\Im \omega|y}\left(\frac{y}{1 + |\omega|y}\right)^{\lambda + \frac{1}{2}}.\end{equation*}
One then shows by induction, as before, that 
\begin{equation*}\left|\tilde{\eta}(\lambda,\omega,s) - \sum_{n = 0}^{N}\eta^1_n(\lambda,\omega,s)\right| \leq C \frac{(CP(s))^N}{N!}e^{|\Im \omega|s}\left(\frac{s}{1 + |\omega|s}\right)^{\lambda + \frac{1}{2}},\end{equation*}
which as $N\to \infty$ yields that $\tilde{\eta} \equiv \eta^1$. Now, to show that $\eta^1$ may be extended continuously to $\omega = 0$, we first rewrite the ODE $\eqref{etaode}$ as 
\begin{equation}\eta''(s) -\frac{\lambda^2 - \frac{1}{4}}{s^2}\eta(s) = \left(\left(\lambda^2 - \frac{1}{4}\right)\left[\frac{1}{r^2T^2} - \frac{1}{s^2} \right] - \frac{1}{rT^2K^2}\left(\frac{T'}{T} + \frac{K'}{K}\right) - \omega^2\right)\eta(s).\end{equation}
The operator on the left-hand side here (with zero boundary conditions at $s = 0$) has the Green's function $H(s-y)\left(s^{\lambda + \frac{1}{2}}y^{-\lambda + \frac{1}{2}} - y^{\lambda + \frac{1}{2}}s^{-\lambda + \frac{1}{2}}\right)$. We thus obtain the integral equation
\begin{equation}\eta^{1,0}(\lambda, \omega, s) = s^{\lambda + \frac{1}{2}} + \int_0^s \left(W(y) - \omega^2\right)\sqrt{sy}\left[\left(\frac{s}{y}\right)^{\lambda} - \left(\frac{y}{s}\right)^{\lambda}\right]\eta^{1,0}(\lambda,\omega,y)dy.\end{equation}
We again solve this via a perturbation series: 
\begin{equation}\eta^{1,0}(\lambda,\omega,s) = \sum_{n=0}^{\infty}\eta^{1,0}_n(\lambda,\omega,s),\end{equation}
where $\eta^{1,0}(\lambda,\omega,s) = s^{\lambda + \frac{1}{2}}$ and 
\begin{equation*}\eta^{1,0}_{n+1}(\lambda,\omega,s) = \int_0^s\left(W(y) - \omega^2\right)\sqrt{sy}\left[\left(\frac{s}{y}\right)^{\lambda} - \left(\frac{y}{s}\right)^{\lambda}\right]\eta^{1,0}_{n}(\lambda,\omega,y)dy.\end{equation*}
For $0<y<s$ we use the obvious bound
\begin{equation*}\left[\left(\frac{s}{y}\right)^{\lambda} - \left(\frac{y}{s}\right)^{\lambda}\right] \leq 2\left(\frac{s}{y}\right)^{\lambda}\end{equation*}
and we can easily show by induction that
\begin{equation}|\eta^{1,0}_{n}(\lambda,\omega,s)| \leq \frac{s^{\lambda + \frac{1}{2}}}{n!\lambda^n}\left(\tilde{P}(s)\right)^n,\end{equation}
where
\begin{equation*}\tilde{P}(s) = \int_0^s \left(|W(y)| + |\omega|^2\right)dy.\end{equation*}
This shows that $\eta^{1,0}$ exists and it is obviously continuous in $\omega$ for small $\omega$. We can follow a similar procedure as above to verify that $\eta^{1,0}$ is analytic in $\omega$ (for any finite $\omega$) for fixed $s >0$ and at least twice continuously differentiable for $s$ for $s >0$. Thus, $\eta^{1,0}$ is a solution of the ODE $\eqref{etaode}$ and, due to the boundary conditions, it also solves the integral equation $\eqref{eta1intequation}$. Using the uniqueness shown above, it follows that $\eta^1 = \eta^{1,0}$. (One might justifiably ask why we bother at all with $\eta^1$. The reason is that the asymptotics for large $\omega$ are imperative to obtain a decay result, but it is difficult to analyze $\eta^{1,0}$ for large $\omega$.) Thus, $\eta^1$ may be extended continuously to $\omega = 0$. We remark that the uniqueness also guarantees that $\eta^1(\lambda,\bar{\omega},s) = \overline{\eta^1(\lambda,\omega,s)}$, and we note that, as can be seen from the construction above, $\eta^{1,0}$ is real-valued for $\omega \in \R$, and hence, $\eta^1$ is real valued for $\omega \in \R$.

\subsection{The Solution with Boundary Conditions at $s = \infty$}
We move on now to construction a solution of $\eqref{etaode}$ satisfying asymptotic boundary conditions as $s \to \infty$; call this solution $\eta^2(\lambda,\omega,s)$. We restrict ourselves for the moment to $\Im \omega \leq 0$, $\omega \neq 0$. Rewriting this ODE again as in $\eqref{etaoderewrite}$, we find the Green's function for the operator on the left-hand side with zero boundary conditions at $ s=\infty$ is given by
\begin{equation}B(\lambda,\omega,s,y) = H(y-s)\frac{i}{2\omega}\left(\eta^2_0(\lambda,\omega,y)\eta^2_0(\lambda,-\omega,s) - \eta^2_0(\lambda,\omega,s)\eta^2_0(\lambda,-\omega,y)\right),\label{defofB}\end{equation}
where
\begin{equation}\eta^2_0(\lambda,\omega,s) = \left(\frac{1}{2}\pi \omega s\right)^{\frac{1}{2}}e^{-\frac{i\pi}{2}\left(\lambda + \frac{1}{2}\right)}H_{\lambda}^{(2)}(\omega s)\label{eta20def}\end{equation}
and $H_{\lambda}^{(2)}$ is the Hankel function of the second kind. Note that $\lim_{s\to \infty}\eta^2_0(\lambda,\omega,s)e^{i\omega s} = 1$. Thus, if we require 
\begin{equation}\lim_{s \to \infty}\eta^2(\lambda,\omega,s)e^{i\omega s} = 1,\label{eta2bc}\end{equation}
then the equivalent integral equation for $\eta^2$ is
\begin{equation}\eta^2(\lambda,\omega,s) = \eta^2_0(\lambda,\omega,s) + \int_s^{\infty}B(\lambda,\omega,s,y)W(y)\eta^2(\lambda,\omega,y)dy.\label{eta2intequation}\end{equation}

We wish to solve this as a perturbation series, so we write 
\begin{equation}\eta^2(\lambda,\omega,s) = \sum_{n=0}^{\infty}\eta^2_{n}(\lambda,\omega,s),\label{eta2series}\end{equation}
where
\begin{equation}\eta^2_{n+1} = \int_s^{\infty}B(\lambda,\omega,s,y)W(y)\eta^2_n(\lambda,\omega,y)dy.\label{eta2ndef}\end{equation}
To address convergence, we note that it is shown in appendix A of \cite{potentialscattering} that for $0<s<y$ we have
\begin{equation}|\eta^2_0(\lambda,\omega,s)| \leq C \left(\frac{|\omega|s}{1 + |\omega|s}\right)^{-\lambda + \frac{1}{2}}e^{(\Im \omega)s}\label{eta20bound}\end{equation}
and
\begin{equation}|B(\lambda,\omega,s,y)| \leq C e^{|\Im \omega| y + (\Im \omega)s}\left(\frac{y}{1 + |\omega| y}\right)^{\lambda + \frac{1}{2}}\left(\frac{s}{1 + |\omega|s}\right)^{-\lambda + \frac{1}{2}}\label{Bbound}\end{equation}
where $C$ depends on $\lambda$. It is easy to show then by induction that
\begin{equation}|\eta_n^2(\lambda,\omega,s)| \leq C \frac{(CQ(s))^n}{n!}\left(\frac{|\omega|s}{1 + |\omega|s}\right)^{-\lambda + \frac{1}{2}}e^{(\Im \omega)s},\label{eta2nbound}\end{equation}
where
\begin{equation}Q(s) = \int_s^{\infty}\frac{y|W(y)|}{1 + |\omega| y}e^{(|\Im \omega| + \Im \omega )y}dy.\label{defofQ}\end{equation}
Note that for $\Im \omega \leq 0$, $Q$ is finite for all $s \in [0,\infty)$ and indeed $\|Q\|_{L^1([0,\infty))} < \infty$, owing to the integrability of $W$ and our requirement that $\Im \omega \leq 0$. Thus $\eta^2$ exists (for $\Im \omega \leq 0$ and $\omega \neq 0$), and the following bounds are obvious
\begin{equation}|\eta^2(\lambda,\omega,s)| \leq Ce^{(\Im \omega)s}\left(\frac{|\omega|s}{1 + |\omega|s}\right)^{-\lambda + \frac{1}{2}}e^{CQ(s)},\label{eta2bound}\end{equation}
and
\begin{equation}|\eta^2(\lambda,\omega,s) - \eta^2_0(\lambda,\omega,s)| \leq Ce^{(\Im \omega)s}\left(\frac{|\omega|s}{1 + |\omega|s}\right)^{-\lambda + \frac{1}{2}}(e^{CQ(s)}-1).\label{eta2minuseta20bound}\end{equation}
Arguments similar to those in the previous section establish smoothness, analyticity, uniqueness, and that $\eta^2$ solves the ODE $\eqref{etaode}$, so we omit the details. We easily obtain the following estimates:
\begin{equation}\left|\frac{d}{ds}\eta^2_0(\lambda,\omega,s)\right| \leq C|\omega| e^{(\Im \omega)s}\left(\frac{|\omega|s}{1 + |\omega|s}\right)^{-\lambda - \frac{1}{2}}\label{deta20bound}\end{equation}
and
\begin{equation}\left|\frac{d}{ds}\eta^2(\lambda,\omega,s) - \frac{d}{ds}\eta^2_0(\lambda,\omega,s)\right| \leq C\left(\frac{|\omega|s}{1 + |\omega|s}\right)^{-\lambda - \frac{1}{2}}e^{(\Im \omega)s}\int_s^{\infty}\left(\frac{|\omega|y}{1 + |\omega|y}\right)^{-\lambda + \frac{1}{2}}e^{CQ(y)}|W(y)|dy.\label{deta2minusdeta20bound}\end{equation}

From $\eqref{eta2bound}$ we see a possible singularity in $\eta^2$ at $\omega = 0$, but this singularity is removable. Indeed, repeating the above construction with the initial function $\eta^{2,0}_0(\lambda,\omega,s) = \omega^{\lambda - \frac{1}{2}}\left(\frac{1}{2}\pi \omega s\right)^{\frac{1}{2}}e^{-\frac{i\pi}{2}\left(\lambda + \frac{1}{2}\right)}H_{\lambda}^{(2)}(\omega s)$ yields a solution $\eta^{2,0}$ of the integral equation
\begin{equation*}\eta^{2,0}(\lambda,\omega,s) = \eta^{2,0}_0(\lambda,\omega,s) + \int_s^{\infty}B(\lambda,\omega,s,y)W(y)\eta^{2,0}(\lambda,\omega,y)dy.\label{eta2intequation}\end{equation*}
This solution satisfies the boundary conditions $\lim_{s\to \infty}\eta^{2,0}(\lambda,\omega,s)e^{i\omega s} = \omega^{\lambda - \frac{1}{2}}$ and it is continuous in $\omega$ up to $\omega = 0$ (from the region $\Im \omega \leq 0$). Finally, $\eta^{2,0}$ can also be obtained from $\eta^2$ in the sense that $\omega^{\lambda - \frac{1}{2}} \eta^2 = \eta^{2,0}$ (by uniqueness).

So we have solved the ODE $\eqref{etaode}$ subject to the boundary conditions $\eqref{eta2bc}$ for $\Im \omega \leq 0$. For $\Im \omega >0$, we obtain a solution of $\eta^2(\lambda,\omega,s)$ of this BVP by defining $\eta^2(\lambda,\omega,s) = \overline{\eta^2(\lambda, \bar{\omega},s)}$. The uniqueness guarantees that this is indeed a solution.

\subsection{Constructing the Resolvent}
We note that for $\Im \omega < 0$, the regularity of $\eta^{1}$ at the origin and the exponential decay of $\eta^2$ as $s\to \infty$ imply that if $\eta^1,\eta^2$ were linearly dependent, then they would produce a nontrivial vector in the kernel of $(\bar{H}_l - \omega)^{-1}$. However, since $\bar{H}_l$ is self-adjoint on a domain in $\mathscr{H}$, the spectrum is real, i.e. $\sigma(\bar{H}_l)\subset \R$ and thus the kernel of $(\bar{H}_l - \omega)^{-1}$ is trivial. Thus, $\eta^1,\eta^2$ must be linearly independent. Since $\eta^1,\eta^2$ solve the same ODE and are linearly independent, we have that the Wronskian $w(\eta^1,\eta^2) \neq 0$ (note also that the Wronskian is easily seen to be independent of $s$). 

What about for $\omega \in \R$? As we noted above, for $\omega \in \R$, $\eta^1$ is real, and more importantly, it has constant phase. However, for $\omega \neq 0$ the boundary conditions $\eqref{eta2bc}$ imply that $\eta^2$ is of variable phase. This implies that $\eta^1,\eta^2$ are linearly independent for real $\omega \neq 0$ and thus that the Wronskian is nonzero for real $\omega \neq 0$. 

For $\omega = 0$, we must argue differently, and we consider the extensions of $\eta^1,\eta^2$ to $\omega = 0$. We recall that, according to the definition $\eqref{defofeta}$, there exist solutions $\gamma^1(\lambda,\omega,u),\gamma^2(\lambda,\omega,u)$ of $\eqref{gammaequation}$ corresponding to $\eta^1(\lambda,\omega,s),\eta^2(\lambda,\omega,s)$, respectively. Let us note also that, using the asymptotics of $u$ described earlier and the definition of $s$, we find 
\begin{equation*}s = -\frac{T(0)}{u}\left(1 + O\left(\frac{1}{u}\right)\right).\end{equation*}
Now let us investigate the asymptotic behavior of $\gamma^1,\gamma^2$. For $\gamma^2$ we have
\begin{align*}1 &= \lim_{s\to \infty} \eta^2(\lambda,\omega,s) \\ &=\lim_{u\nearrow 0} r(u) \gamma^2(\lambda,\omega,u) \\&= \lim_{u\nearrow 0} \left(-\frac{1}{u} + O(1)\right)\gamma^2(\lambda,\omega,u).\end{align*}
This implies that $\gamma^2$ decays as $u\nearrow 0$. For $\gamma^1$ we have
\begin{align*}1 &= \lim_{s\to 0} s^{-\lambda-\frac{1}{2}}\eta^1(\lambda,\omega,s) \\ &= \lim_{u\to -\infty}\left(\frac{-u}{T(0)}\right)^{\lambda + \frac{1}{2}}\left( 1 + O\left(\frac{1}{u}\right)\right)r(u)\gamma^1(\lambda,\omega,u)\\ &= \lim_{u\to -\infty}\left(\frac{-u}{T(0)}\right)^{\lambda + \frac{1}{2}}\left( 1 + O\left(\frac{1}{u}\right)\right)\left(-\frac{1}{u} + O\left(\frac{1}{u^2}\right)\right)\gamma^1(\lambda,\omega,u).\end{align*}
This implies that $\gamma^1$ either decays as $u \to -\infty$ or tends to a constant (depending on $\lambda$). From $\eqref{gammaequation}$ with $\omega = 0$, we see that $\gamma^1$ and $\gamma^2$ are either strictly concave or convex. Thus, $\gamma^1$ and $\gamma^2$ must be linearly independent. In particular, since they solve the same ODE, the Wronskian $w(\gamma^1,\gamma^2) \neq 0$. Furthermore, an easy calculation shows that $w(\eta^1,\eta^2) = w(\gamma^1,\gamma^2)$. We have thus shown that $w(\gamma^1(\lambda,\omega,u),\gamma^2(\lambda,\omega,u))$ is never zero.

Thus, the function $h(\omega,u,v)$ defined by
\begin{equation} h(\omega,u,v) = -\frac{1}{w(\gamma^1(\lambda,\omega,u),\gamma^2(\lambda,\omega,u))}\begin{cases}\gamma^1(\lambda,\omega,u)\gamma^2(\lambda,\omega,v), & u \leq v \\ \gamma^1(\lambda,\omega,v)\gamma^2(\lambda,\omega,u), & v < u\end{cases}\label{defofh}\end{equation}
is well-defined. Note that since we are considering a fixed mode, we omit the functional dependence of $\lambda$ in $h$. Now, it's clear that $h$ is continuous in $u,v$ for fixed $\omega \in \C$, but moreover, $h$ is also continuous in $\omega$ over all of $\C$ for fixed $u,v$. The only possible difficulty comes near $\omega = 0$. But, notice that $h$ is unchanged if we consider $\omega^{\lambda - \frac{1}{2}}\gamma^2$ instead of $\gamma^2$ and the continuity follows. We next claim that $h$ multiplied by the operator in $\eqref{gammaequation}$ ``acts like the Dirac functional''. More precisely,
\begin{prop}The function $h(\omega,u,v)$ defined in $\eqref{defofh}$ satisfies
\begin{equation}\int_{-\infty}^0 h(\omega,u,v) \left(-\D_v^2 - r^4(v)\omega^2 + \frac{r^2}{T^2}l(l+1)\right)\gamma(v)dv = \gamma(u)\label{hintequation}\end{equation}
for any $\gamma \in C_0^{\infty}(-\infty,0)$.\end{prop}
\begin{proof}This follows from a simple calculation where we split the integral into $\int_{-\infty}^u$ and $\int_u^0$ and integrate by parts in these regions (since $h$ is smooth in $v$ in these regions).\end{proof}

Let us now compute the resolvent. To this end, we define integral operator $S_{\omega}$ acting on the domain $\mathscr{D}(S_{\omega}) = \left\{(\bar{H}_l - \omega)\Gamma : \Gamma \in \mathscr{C}_l^2\right\}$, with $S_{\omega}\Phi$ being given by
\begin{equation}(S_{\omega}\Phi)(u) = \int_{-\infty}^0\left[\delta(u-v)\left(\begin{array}{cc}0 & 0 \\ 1 & 0\end{array}\right) + r^4(v)h(\omega,u,v)\left(\begin{array}{cc}\omega & 1 \\ \omega^2 & \omega \end{array}\right)\right]\Phi(v)dv.\label{defofS}\end{equation}
We next claim that, in fact, $S_{\omega} = (\bar{H}_l - \omega)^{-1}$ on $\mathscr{H}$. To see this, first note that $\mathscr{D}(S_{\omega})$ is dense in $\mathscr{H}$. This can be argued easily, and indeed, the argument is identical to the one presented in \cite{thecauchyproblemforthewaveequationintheschwarzschildgeometry}. However, we can go further; we claim that, in fact, the set $\left\{ (\bar{H}_l - \omega)\Gamma : \Gamma \in C_0^{\infty}(-\infty,0)^2\right\}$ is dense in $\mathscr{D}(S_{\omega})$ and thus dense in $\mathscr{H}$. To prove it, we must show that for each $\Phi \in \mathscr{C}^2$ there exists a sequence $\Phi_n \in C_0^{\infty}(-\infty,0)^2$ such that $(H_l - \omega)(\Phi_n - \Phi) \to 0$ as $n \to \infty$ (in the norm $\|\cdot\|$ induced by the inner product $\langle \cdot, \cdot \rangle_l$). Recalling the specific form of $H_l$ in $\eqref{defofHl}$, we compute for $\Phi = (\phi_1,\phi_2)^T \in \mathscr{C}^2$,
\begin{align*}&\|\left(\bar{H}_l - \omega\right)\Phi\|^2 \\&= \|(H_l - \omega)\Phi\|^2 \\ &= \int_{-\infty}^0r^4|A_l \phi_1 - \omega\phi_2|^2 + |\D_u\phi_2 - \omega\D_u\phi_1|^2 + \frac{r^2}{T^2}l(l+1)|\phi_2 - \omega\phi_1|^2du \\ &\leq 2\int_{-\infty}^0r^4|A_l\phi_1|^2 + r^4|\omega|^2|\phi_2|^2 + |\D_u\phi_2|^2 + |\omega|^2|\D_u\phi_1|^2 + \frac{r^2}{T^2}l(l+1)|\phi_2|^2 +\frac{r^2}{T^2}l(l+1)|\omega|^2|\phi_1|^2du.\end{align*}
Now, if we define $\psi_i(r) = \phi_i(u(r))$, then we get $\D_r \psi_i(r)\frac{r^2}{KT} = \D_u \phi_i(u(r))$, owing to the definition of $u(r)$ in $\eqref{ucoordinate}$. We also find $\D_u^2\phi_i(u(r)) = \frac{r^2}{KT}\D_r\left(\D_r\psi_i(r)\frac{r^2}{KT}\right)$. Now plugging in the specific form of $A_l$ and changing from the $u$ variable to the $r$ variable, we find
\begin{align*}\|(\bar{H}_l - \omega)\|^2 \leq 2\int_0^{\infty}\frac{2}{KTr^2}&\left|\D_r \left(\D_r \psi_1\frac{r^2}{KT}\right)\right|^2 + \frac{2K(l(l+1))^2}{T^3r^2}|\psi_1|^2 + KTr^2|\omega|^2|\psi_2|^2 + \frac{r^2}{KT}|\D_r\psi_2|^2 \\&+ |\omega|^2\frac{r^2}{KT}|\D_r\psi_1|^2 + \frac{Kl(l+1)}{T}|\psi_2|^2 + \frac{K}{T}l(l+1)|\omega|^2|\psi_1|^2 dr\end{align*}
The first two terms in the above integrand are troublesome. Let us look first at the second term for $l\neq 0$. First we recall that $\phi_1 \in \mathscr{C}_l$ implies that $\psi_1$ vanishes outside of a large ball (say of radius $R$). Thus, for some $r_0 >0$, we have
\begin{align*}\int_0^{\infty}\frac{|\psi_1|^2}{r^2}dr &= \int_0^{R}\frac{|\psi_1|^2}{r^2}dr \\ &= \int_0^{r_0}\frac{|\psi_1|^2}{r^2}dr + \int_{r_0}^R\frac{|\psi_1|^2}{r^2}dr \\
&\leq \frac{1}{r_0^2}\int_0^{\infty}|\psi_1|^2dr + C \int_0^{r_0}|\psi_1| dr, \text{ since } \psi_1 = O(r^2) \text{ near } r = 0 \\
&\leq \frac{1}{r_0}^2\int_0^{\infty}|\psi_1|^2 dr + 2Cr_0^2\int_0^{r_0}|\psi_1|^2 dr, \text{ using H\"{o}lder}\\
&\leq C\int_0^{\infty}|\psi_1|^2dr.\end{align*}
Thus, this term is actually bounded. To see that the first term is no problem, we simply observe that the absolute value will have at least an $r$ inside it, and this will come out and cancel the $r^2$ in the denominator. This, using that $T,K$ are bounded and bounded away from zero, as well as the fact that $\psi_1,\psi_2$ vanish outside of a large ball, we can bound the above by the $H^2$ norms of $\psi_1,\psi_2$. More precisely, 
\begin{equation*}\|(H_l - \omega)\Phi\|^2 \leq C(1 + |\omega|^2)\|\Psi\|_{H^2(0,\infty)^2},\end{equation*}
where $C$ depends on $l$ and the support of $\Psi$. Now, since $C_0^{\infty}(0,\infty)^2$ is dense in $H^2(0,\infty)^2$, this shows that for $(\Psi_n)\subset C_0^{\infty}(0,\infty)^2$ with $\|\Psi - \Psi_n\| \to 0$ as $n \to \infty$, we can find a sequence $(\Phi_n) = (\Psi(r(u))_n)$ so that $(\Phi_n)\subset C_0^{\infty}(-\infty,0)^2$ and $\|(H_l - \omega)(\Phi - \Phi_n)\| \to 0$ as $n \to \infty$. Therefore, given $\Phi\in \mathscr{C}_l^2$, we consider $\Psi(r) := \Phi(u(r))$. Then $\Psi$ is surely in $H^2(0,\infty)^2$, and we may therefore find a sequence $(\Psi_n)\subset C_0^{\infty}(0,\infty)^2$ so that $\Psi_n \to \Psi$ as $n \to \infty$. Then define a sequence $(\Phi_n) \subset C_0^{\infty}(-\infty,0)$ by $\Phi_n(u) = \Psi_n(r(u))$. By the above, we know that $(H_l - \omega)(\Phi - \Phi_n) \to 0$ as $n \to \infty$. This then proves our claim; i.e. $\left\{(\bar{H}_l - \omega)\Gamma : \Gamma \in C_0^{\infty}(-\infty,0)^2\right\}$ is dense in the set $\left\{(\bar{H}_l - \omega)\Gamma : \Gamma \in \mathscr{C}_l^2\right\}$.

Now using equation $\eqref{hintequation}$, it is easy to check that for $\Psi \in C_0^{\infty}(-\infty,0)$, we have $(S_{\omega}(\bar{H}_l - \omega)\Psi)(u) = \Psi(u)$. In other words, $S_{\omega} = (\bar{H}_l - \omega)^{-1}$ on $\left\{ (\bar{H}_l - \omega)\Gamma : \Gamma \in C_0^{\infty}(-\infty,0)^2\right\}$. But since the resolvent is a bounded operator and $\left\{ (\bar{H}_l - \omega)\Gamma : \Gamma \in C_0^{\infty}(-\infty,0)^2\right\}$ is dense in $\mathscr{D}(S_{\omega})$, which is dense in $\mathscr{H}$, it follows that $S_{\omega} = (\bar{H}_l - \omega)^{-1}$ on $\mathscr{H}$. 

Now, according to Stone's formula (Theorem~\ref{stoneformula}), if we let $k(\omega,u,v)$ denote the kernel of the operator $S_{\omega}$, then for any $\Psi \in \mathscr{H}$ we have
\begin{equation*}\frac{1}{2}\left(P_{[a,b]} + P_{(a,b)}\right) \Psi(u) = \lim_{\eps\searrow 0}\frac{1}{2\pi i}\int_a^b \int_{-\infty}^0(k(\omega + i\eps,u,v) - k(\omega - i\eps,u,v))\Psi(v)dvd\omega.\end{equation*}
Recalling that $\eta^i(\lambda,\bar{\omega},s) = \overline{\eta^i(\lambda,\omega,s)}$ and noting that the same must therefore hold for the $\gamma^i$, this implies that $h(\omega + i\eps,u,v) = \overline{h(\omega - i\eps,u,v)}$ for $\omega \in \R$, and thus $k(\omega + i\eps,u,v) =\overline{k(\omega - i\eps,u,v)}$. This gives
\begin{equation}\frac{1}{2}\left(P_{[a,b]} + P_{(a,b)}\right) \Psi(u) = \lim_{\eps\searrow}-\frac{1}{\pi}\int_a^b\int_{-\infty}^0\text{Im}(k(\omega - i\eps,u,v))\Psi(v)dv\label{spectralrep1}\end{equation}
and we note again that this converges in the $\mathscr{H}$-norm. In particular, we would like to derive a spectral representation for the data $\Psi_0^{lm}$. We would like to consider the representation in $\eqref{spectralrep1}$ and interchange the limit and the integral, so we must analyze $\text{Im}(k(\omega - i\eps,u,v))$. Indeed, we know that by the above, at the worst $h(\omega - i\eps,u,v)$ tends to a constant at $u = -\infty$ (this follows from the discussion above on $\gamma^1$). But there is a factor of $r^4$ in $\text{Im}(k)$ to enforce decay. Indeed, as $u \to -\infty$, $r^4 = O\left(\frac{1}{u^4}\right)$. Since $\Psi^{lm}_0$ tends to a constant at $u = -\infty$, we see that we are justified in switching the order of the limit and the integration (also using, of course, the continuity of $\Im k$), for fixed $u$. From the norm convergence implied in Stone's formula, we thus obtain the spectral representation of $\Psi^{lm}_0$:
\begin{equation}\frac{1}{2}\left(P_{[a,b]} + P_{(a,b)}\right) \Psi^{lm}_0(u) = -\frac{1}{\pi}\int_a^b\int_{-\infty}^0 \text{Im}(k(\omega,u,v))\Psi^{lm}_0(v)dvd\omega.\label{specrep2}\end{equation}
This yields that $P_{\{a\}} = 0$ for any $a\in \R$ and that the spectrum of $\bar{H}_l$ is absolutely continuous. Thus we have
\begin{equation}P_{(a,b)}\Psi_0^{lm}(u) = -\frac{1}{\pi}\int_a^b\int_{-\infty}^0 \text{Im}(k(\omega,u,v))\Psi^{lm}_0(v)dvd\omega.\label{specrep3}\end{equation}
Finally, using the spectral theorem and the fact that $e^{-it\bar{H}_l}$ is unitary, we derive the representation for $\Psi^{lm}(t,u)$:
\begin{equation}\Psi^{lm}(t,u) = -\frac{1}{\pi}\int_{\R}e^{-i\omega t}\int_{-\infty}^0 \text{Im}(k(\omega,u,v))\Psi^{lm}_0(v)dvd\omega.\label{repform}\end{equation}

\section{Decay}To show that the solution $\Psi^{lm}$ decays, we would like to use use the Riemann-Lebesgue lemma and the representation formula $\eqref{repform}$. In particular, if we show that the integrand within the $\omega$-integral is in $L^1(\R,\C^2)$, then the Riemann-Lebesgue lemma guarantees that for fixed $u \in (-\infty,u)$, $\Psi^{lm}(t,u) \to 0$ as $t \to \infty$. To this end, let us find a more useful form of the integrand. First, we claim that the pair $\{\eta^2,\overline{\eta^2}\}$ forms a fundamental set for the ODE $\eqref{etaode}$ for $\omega \in \R\setminus\{0\}$. To verify this, we first compute the Wronskian $w(\eta^2_0,\overline{\eta^2_0})$. An easy calculation shows that $w(\eta^2_0,\overline{\eta^2_0}) = 2i\neq 0$. Now, $\eqref{eta2minuseta20bound}$ implies that the difference between $\eta^2$ and $\eta^2_0$ tends to zero as $s\to\infty$, which means (invoking $\eqref{deta2minusdeta20bound}$ as well) that we must have 
\begin{equation}w(\eta^2,\overline{\eta^2}) = w(\eta^2_0, \overline{\eta^2_0})= 2i\neq 0,\label{computewronskian}\end{equation}
since $w(\eta^2,\overline{\eta^2})$ is constant in $s$. Thus, the pair $\{\eta^2,\overline{\eta^2}\}$ forms a fundamental set for $\omega \in \R\setminus\{0\}$. This implies that $\{\gamma^2,\overline{\gamma^2}\}$ forms a fundamental set for $\eqref{gammaequation}$ for $\omega \in \R\setminus\{0\}$. Thus, there exist numbers (depending only on $\omega$) $c(\omega),d(\omega)$ such that 
\begin{equation}\gamma^1(\lambda,\omega,u) = c(\omega)\gamma^2(\lambda,\omega,u) + d(\omega)\overline{\gamma^2(\lambda,\omega,u)},\label{gamma1bygamma2}\end{equation}
and where we know that $d(\omega) \neq 0$ for all $\omega$. Note then that $\eqref{computewronskian}$ implies that $w(\gamma^1,\gamma^2) = -2id(\omega)$ and $w(\gamma^1,\overline{\gamma^2}) = 2ic(\omega)$.

Next, we let $\phi^1_{\omega} = \Re \gamma^2, \phi^2_{\omega} = \Im \gamma^2$, and $\Phi^a_{\omega} = (\phi^a_{\omega},\omega\phi^a_{\omega})^T$ (note that we are dropping the $\lambda$ argument, since for our purposes it is superfluous, and we denote the $\omega$ dependence by a subscript). A short calculation then shows 
\begin{equation}\Im h_{\omega}(u,v) = \frac{1}{2\omega}\sum_{a,b = 1}^2\alpha_{ab}\phi^a_{\omega}(u)\phi^b_{\omega}(v)\label{simplifiedh}\end{equation}
where
\begin{equation}\alpha_{11} = 1 + \text{Re}\left(\frac{c}{d}\right), \alpha_{22} =  1 - \text{Re}\left(\frac{c}{d}\right), \alpha_{12} = \alpha_{21} = -\text{Im}\left(\frac{c}{d}\right).\label{alphadef}\end{equation}
Now if we write $\Psi^{lm}_0 = (\psi^{lm}_{0,1}, \psi^{lm}_{0,2})^T$, we have
\begin{align*}\int_{-\infty}^0 \text{Im}(k_{\omega}(u,v))\Psi^{lm}(v)dv &= \frac{1}{2\omega^2}\sum_{a,b = 1}^2 \alpha_{ab}\Phi^a_{\omega}(u)\int_{-\infty}^0 r^4 \phi^b_{\omega}(v)(\omega^2 \psi^{lm}_{0,1} + \omega \psi^{lm}_{2,0})dv.\end{align*}
But now let us use the fact that $\phi^b(v)$ satisfies $-\D_v^2 \phi^b_{\omega} + \frac{r^2}{T^2}l(l+1)\phi^b_{\omega} =  \omega^2 r^4\phi^b_{\omega}$. We plug this in the above integral to obtain
\begin{align*}\int_{-\infty}^0 \text{Im}(k_{\omega}(u,v))\Psi^{lm}_0(v)dv &= \frac{1}{2\omega^2}\sum_{a,b = 1}^2\alpha_{ab}\Phi^a_{\omega}(u)\int_{-\infty}^0 \psi^{lm}_{0,1}\left(-\D_v^2 + \frac{r^2}{T^2}l(l+1)\right)\phi^b_{\omega} + \omega r^4 \psi^{lm}_{0,2}\phi^b_{\omega}dv.\end{align*}
We now introduce the additional assumption that $\Psi^{lm}_0 \in C_0^{\infty}(-\infty,0)^2$.\footnote{This corresponds to assuming that the data in problem $\eqref{cp2}$ is supported away from the origin. But note that, to work in the $u$ coordinate as we have done, which maps the interval $(0,\infty)$ to $(-\infty,0)$, with $r = 0$ corresponding to $u = -\infty$, this does not seem like an unreasonable requirement.} Owing to this assumption, we may integrate by parts in the above integral and obtain
\begin{align}\int_{-\infty}^0 \text{Im}(h_{\omega}(u,v))\Psi^{lm}_0(v)dv &= \frac{1}{2\omega^2}\sum_{a,b = 1}^2\Phi^a_{\omega}(u) \int_{-\infty}^0 (\D_v \psi^{lm}_{0,1})(\D_v\phi^b_{\omega}) + \frac{r^2}{T^2}l(l+1)\phi^b_{\omega}\psi^{lm}_{0,1} + \omega r^4 \psi^{lm}_{0,2}\phi^b_{\omega}dv\notag \\
&= \frac{1}{2\omega^2}\sum_{a,b = 1}^2\alpha_{ab}\Phi^a_{\omega}(u)\langle \Psi^{lm}_0, \Phi^b_{\omega}\rangle_l,\end{align}
and so
\begin{equation}\Psi^{lm}(t,u) = \frac{1}{2\pi} \int_{\R} e^{-i\omega t} \frac{1}{\omega^2}\sum_{a,b = 1}^2\alpha_{ab}(\omega)\Phi^a_{\omega}(u)\langle \Psi^{lm}_0, \Phi^b_{\omega}\rangle_ld\omega.\label{repform2}\end{equation}
We have already demonstrated that the integrand above is continuous in $\omega$, so to show the integrand is in $L^1(\R,\C^2)$, we need only to analyze it for $|\omega|\gg 1$. First we recall the formulas for $c(\omega),d(\omega)$: $w(\gamma^1,\gamma^2) = -2id(\omega), w(\gamma^1,\overline{\gamma^2}) = 2ic(\omega)$. Let us fix $s = s_0 \in (0,\infty)$ and we will compute $w(\eta^1,\eta^2)(s)$. Indeed, recalling the bounds $\eqref{deta2minusdeta20bound},\eqref{eta2minuseta20bound}, \eqref{deta1minusdeta10bound}, \eqref{eta1minuseta10bound}$ and considering $\omega \in \R$, we have

\begin{equation}\left|\eta^1(\lambda,\omega,s) - \eta^1_0(\lambda,\omega,s) \right| = O\left(\frac{1}{|\omega|^{\lambda + \frac{3}{2}}}\right),\end{equation}
\begin{equation}\left| \frac{d}{ds}\eta^1(\lambda,\omega,s) - \frac{d}{ds}\eta^1_0(\lambda,\omega,s)\right| = O\left(\frac{1}{|\omega|^{\lambda + \frac{1}{2}}}\right),\end{equation}
\begin{equation}\left| \eta^2\lambda,\omega,s) - \eta^2_0(\lambda,\omega,s) \right| = O\left(\frac{1}{|\omega|}\right),\end{equation}
and 
\begin{equation}\left| \frac{d}{ds}\eta^2(\lambda,\omega,s) - \frac{d}{ds}\eta^2_0(\lambda,\omega,s) \right| = O(1).\end{equation}
Thus we have $w(\eta^1,\eta^2)  = w(\eta^1_0,\eta^2_0) + O\left(\frac{1}{|\omega|^{\lambda + \frac{1}{2}}}\right)$. Then an easy calculation shows that $w(\eta^1_0, \eta^2_0) = O\left(\frac{1}{|\omega|^{\lambda + \frac{1}{2}}}\right)$, which implies that $w(\eta^1,\eta^2) = O\left(\frac{1}{|\omega|^{\lambda + \frac{1}{2}}}\right)$. We can show similarly that $w(\eta^1,\overline{\eta^2}) = O\left(\frac{1}{|\omega|^{\lambda + \frac{1}{2}}}\right)$. Thus, we have that $c,d = O\left(\frac{1}{|\omega|^{\lambda + \frac{1}{2}}}\right)$. This implies that
\begin{equation} |\alpha_{ab}| \leq 1 + O\left(|\omega|^{\lambda + \frac{1}{2}}\right).\end{equation}

Next, we note that $|\phi^a_{\omega}| \leq |\gamma^2|$. But from the bound $\eqref{eta2bound}$, considered for $\omega \in \R$, we have that (considering $s = s(u)$) $\gamma^2(\lambda,s,u) = O(1)$ for large $\omega$. Finally, we look at the term $\langle \Psi^{lm}_0,\Phi^b_{\omega}\rangle_l$. We have
\begin{align*}\langle\Psi^{lm}_0,\Phi^b_{\omega}\rangle_l &= \int_{-\infty}^0 (\D_v\psi^{lm}_{0,1})(\D_v\phi^b_{\omega}) + \frac{r^2}{T^2}l(l+1)\gamma^b_{\omega} \psi^{lm}_{0,1} + \omega r^4 \psi^{lm}_{0,2}\phi^b_{\omega} dv \\
&= \int_{-\infty}^0 (-\D_v^2 \psi^{lm}_{1,0} + \omega r^4 \psi^2_{lm} + \frac{r^2}{T^2}l(l+1) \psi^{lm}_{1,0})\phi^b_{\omega}dv \\
&= \frac{1}{\omega^2}\int_{-\infty}^0 (-\D_v^2 \psi^{lm}_{1,0} + \omega r^4 \psi^2_{lm} + \frac{r^2}{T^2}l(l+1) \psi^{lm}_{1,0})\left(-\D_v^2 \phi^b_{\omega} + \frac{r^2}{T^2} + \frac{r^2}{T^2}l(l+1)\phi^b_{\omega}\right)\frac{1}{r^4}dv,\end{align*}
where we used that $-\D_v^2 \phi^b_{\omega} + \frac{r^2}{T^2}l(l+1)\phi^b_{\omega} =  \omega^2 r^4\phi^b_{\omega}$. Integrating by parts and iterating this argument as many times as we please, we obtain \emph{arbitrary} polynomial decay in $\omega$. This polynomial decay is enough to then guarantee that the integrand in $\eqref{repform2}$ is in $L^1(\R,\C^2)$, and then by the Riemann-Lebesgue lemma, we are assured that $\Psi^{lm}(t,u) \to 0$ for fixed $u$ as $t \to \infty$. That the modal decay implies decay of the full solution $\Psi$ follows exactly as in \cite{thecauchyproblemforthewaveequationintheschwarzschildgeometry}. Translating this back into the $r$-coordinate, this implies that for fixed $r \in [0,\infty)$, the solution $\zeta$ of $\eqref{cp2}$, under the additional requirement that $Z_0 \in C_0^{\infty}(R^3\setminus\{0\})$ decays as $t\to \infty$. Thus we have the following theorem:
\begin{thm}Consider problem $\eqref{cp2}$ in a particle-like geometry. If the data is smooth and compactly supported away from the origin, then the solution decays in $L^{\infty}_{\text{loc}}$ as $t\to\infty$.\end{thm}

\section{Application to Particle-like Solutions of EYM}
Finally, we note that particle-like solutions of the SU(2) EYM equations satisfy the conditions $\eqref{firstcondn} - \eqref{TKasymptotics}$, c.f. \cite{existenceofinifinitelymanysmoothstaticglobalsolutionsoftheeinsteinyangmillsequations}. The behavior at the origin follows by simple Taylor expansions and the far-field behavior follows from the results in \cite{existenceofinifinitelymanysmoothstaticglobalsolutionsoftheeinsteinyangmillsequations} (with $K^2 = A^{-1}$ and an asymptotic expansion of the metric coefficients at infinity. Thus, solutions of the wave equation in SU(2) EYM particlelike geometry, with data that is smooth and compactly supported away from the origin, decay in $^{\infty}_{\text{loc}}$ as $t\to \infty$.

\section{Acknowledgements}We thank Volker Elling for his helpful comments in discovering the right boundary conditions to impose at $r = 0$, and the author would like to thank especially his advisor Joel Smoller for introducing this problem, for his many helpful discussions, and for his financial support, NSF contract no. DMS-0603754.
\bibliographystyle{plain}
\bibliography{mybibliography}

\begin{thebibliography}{10}

\bibitem{hiddensymmetriesanddecayforthewaveequationonthekerrspacetime}
L.~{Andersson} and P.~{Blue}.
\newblock {Hidden symmetries and decay for the wave equation on the Kerr
  spacetime}.
\newblock {\em ArXiv e-prints}, August 2009.

\bibitem{courant1}
R.~Courant and D.~Hilbert.
\newblock {\em Methods of Mathematical Physics}, volume~1.
\newblock Interscience Publishers, Inc., New York, NY, 1953.

\bibitem{aproofoftheuniformboundednessofsolutionstothewaveequationsonslowlyrot%
atingkerrbackgrounds}
M.~{Dafermos} and I.~{Rodnianski}.
\newblock {A proof of the uniform boundedness of solutions to the wave equation
  on slowly rotating Kerr backgrounds}.
\newblock {\em ArXiv e-prints}, May 2008.

\bibitem{potentialscattering}
V.~De~Alfaro and T.~Regge.
\newblock {\em Potential Scattering}.
\newblock Amsterdam, 1965.

\bibitem{onpointwisedecayoflinearwavesonaSchwarzschildblackholebackground}
R.~{Donninger}, W.~{Schlag}, and A.~{Soffer}.
\newblock {On pointwise decay of linear waves on a Schwarzschild black hole
  background}.
\newblock {\em ArXiv e-prints}, November 2009.

\bibitem{decayofsolutionsofthewaveequationinthekerrgeometry}
F.~{Finster}, N.~{Kamran}, J.~{Smoller}, and {S.-T.} {Yau}.
\newblock {Decay of solutions of the wave equation in the Kerr geometry}.
\newblock {\em Communications in Mathematical Physics}, 264:465--503, June
  2006.

\bibitem{fritz}
F.~John.
\newblock {\em Partial Differential Equations}.
\newblock Number~1 in Applied Mathematical Sciences. Springer-Verlag, New York,
  NY, 4th edition, 1982.

\bibitem{thecauchyproblemforthewaveequationintheschwarzschildgeometry}
J.~Kronthaler.
\newblock The cauchy problem for the wave equation in the schwarzschild
  geometry.
\newblock {\em Journal of Mathematical Physics}, 47(4):042501, 2006.

\bibitem{decayratesforsphericalscalarwavesintheSchwarzschildgeometry}
J.~{Kronthaler}.
\newblock {Decay rates for spherical scalar waves in the Schwarzschild
  geometry}.
\newblock {\em ArXiv e-prints}, September 2007.

\bibitem{reedsimon1}
M.~Reed and B.~Simon.
\newblock {\em {Functional Analysis, Volume 1 (Methods of Modern Mathematical
  Physics)}}.
\newblock Academic Press, January 1981.

\bibitem{existenceofinifinitelymanysmoothstaticglobalsolutionsoftheeinsteinyan%
gmillsequations}
J.~A. Smoller and A.~G. Wasserman.
\newblock Existence of infinitely-many smooth, static, global solutions of the
  {E}instein/{Y}ang-{M}ills equations.
\newblock {\em Communicatons in Mathematial Physics}, 151.

\bibitem{specialfunctions}
Z.~Wang and D.~R. Guo.
\newblock {\em Special Functions}.
\newblock World Scientific, Singapore ; New Jersey, 1989.

\bibitem{watsonbesselfunctions}
G.~N. Watson.
\newblock {\em A Treatise on the Theory of Bessel Functions}.
\newblock The University Press; The Macmillan Company, Cambridge {$[$}Eng.{$]$}
  New York, 1944.

\end{thebibliography}
\end{document}